\newif\iflong
\newif\ifshort
\newif\ifreference

\longtrue

\iflong \else \shorttrue \fi

\documentclass[runningheads]{llncs}

\usepackage[utf8]{inputenc}
\usepackage{amsmath}    
\usepackage{amssymb}    
\usepackage{mathtools}  
\usepackage{microtype}  
\usepackage{multirow}   
\usepackage{booktabs}   
\usepackage{subcaption} 
\usepackage[ruled,linesnumbered]{algorithm2e} 
\usepackage[usenames,dvipsnames,table]{xcolor} 
\usepackage{tikz}
\usetikzlibrary{calc}
\usepackage{nag}       
\usepackage{todonotes} 
\usepackage{tcolorbox} 
\usepackage{enumitem}
\usepackage{soul}
\usepackage[hidelinks]{hyperref}  

\newcommand{\qedhere}{\tag*{$\square$}}
\newcommand{\vc}{\mathbf{vcn}}
\newcommand{\fe}{\mathbf{fen}}
\newcommand{\nd}{\mathbf{nd}}
\newcommand{\BRAC}{\textsc{BRAC}}
\newcommand{\cc}[1]{{\mbox{\textnormal{\textsf{#1}}}}\xspace} 
\newcommand{\NP}{\cc{NP}}
\newcommand{\FPT}{\cc{FPT}}

\newcommand{\R}{\mathbb{R}}
\newcommand{\bigO}{\mathcal{O}}
\newcommand{\yesInstance}{yes-instance}
\newcommand{\noInstance}{no-instance}

\newcommand{\noInstances}{no-instances}

\title{Fixed-Parameter Algorithms for Computing RAC Drawings of Graphs}
\titlerunning{Fixed-Parameter Algorithms for Computing RAC Drawings of Graphs}
\author{Cornelius Brand\orcidID{0000-0002-1929-055X} \and Robert Ganian\orcidID{0000-0002-7762-8045} \and Sebastian Röder \and Florian Schager}
\institute{Algorithms and Complexity Group, TU Wien (Vienna University of Technology) \\
\email{\{cbrand, rganian\}@ac.tuwien.ac.at, sebastian.roeder@student.tuwien.ac.at, florian.schager@tuwien.ac.at}
}

\usepackage{nameref,cleveref}
\Crefname{splemma}{Lemma}{Lemmas}
\Crefname{sptheorem}{Theorem}{Theorems}
\Crefname{spdefinition}{Definition}{Definitions}
\Crefname{spproperty}{Property}{Properties}
\Crefname{spcorollary}{Corollary}{Corollaries}

\begin{document}

\spnewtheorem{sptheorem}{Theorem}{\bfseries}{\itshape}
\spnewtheorem{splemma}[sptheorem]{Lemma}{\bfseries}{\itshape}
\spnewtheorem{spproperty}[sptheorem]{Property}{\bfseries}{\itshape}
\spnewtheorem{spdefinition}[sptheorem]{Definition}{\bfseries}{\itshape}
\spnewtheorem{spobservation}[sptheorem]{Observation}{\bfseries}{\itshape}
\spnewtheorem{spcorollary}[sptheorem]{Corollary}{\bfseries}{\itshape}

\renewenvironment{theorem}{\begin{sptheorem}}{\end{sptheorem}}
\renewenvironment{lemma}{\begin{splemma}}{\end{splemma}}
\renewenvironment{property}{\begin{spproperty}}{\end{spproperty}}
\renewenvironment{definition}{\begin{spdefinition}}{\end{spdefinition}}
\newenvironment{observation}{\begin{spobservation}}{\end{spobservation}}
\renewenvironment{corollary}{\begin{spcorollary}}{\end{spcorollary}}

\definecolor{turquoise}{rgb}{0.19, 0.84, 0.78}

\maketitle

\begin{abstract}

In a right-angle crossing (RAC) drawing of a graph, each edge is represented as a polyline and edge crossings must occur at an angle of exactly $90^\circ$, where the number of bends on such polylines is typically restricted in some way. 
While structural and topological properties of RAC drawings have been the focus of extensive research, little was known about the boundaries of tractability for computing such drawings.
In this paper, we initiate the study of RAC drawings from the viewpoint of parameterized complexity. 
In particular, we establish that computing a RAC drawing of an input graph $G$ with at most $b$ bends (or determining that none exists) is fixed-parameter tractable parameterized by either the feedback edge number of $G$, or $b$ plus the vertex cover number of $G$.
\end{abstract}

\section{Introduction}

Today we have access to a wealth of approaches and tools that can be used to draw planar graphs, including, e.g., F\'ary's Theorem~\cite{Fary48} which guarantees the existence of a planar straight-line drawing for every planar graph and the classical algorithm of Fraysseix, Pach and Pollack~\cite{FraysseixPP88} that allows us to obtain straight-line planar drawings on an integer grid of quadratic size. 
However, much less is known about the kinds of drawings that can be achieved for non-planar graphs. 
The study of combinatorial and algorithmic aspects of such drawings lies at the heart of a research direction informally referred to as ``beyond planarity'' (see, e.g., the relevant survey and book chapter~\cite{DidimoLM19,Didimo20}).

An obvious goal when attempting to visualize non-planar graphs would be to obtain a drawing which minimizes the total number of crossings. 
This question is widely studied within the context of the crossing number of graphs, and while obtaining such a drawing is \NP-hard~\cite{CrossingNPh83} it is known to be fixed-parameter tractable when parameterized by the total number of crossings required thanks to a seminal result of Grohe~\cite{Grohe04}. 
However, research over the past twenty years has shown that drawings which minimize the total number of crossings are not necessarily optimal in terms of human readability. 
Indeed, the topological and geometric properties of such drawings may have a significantly larger impact than the total number of crossings, as was observed, e.g., by the initial informal experiment of Mutzel~\cite{Mutzel01} and the pioneering set of user experiments carried out by the graph drawing research lab at the University of Sydney~\cite{Huang07,HuangHE08,HuangEH14}. 
The latter works demonstrated that ``large-angle drawings'' (where edge crossings have larger angles) are significantly easier to read than drawings where crossings occur at acute angles.

Motivated by these findings, in 2011 Didimo, Eades, and Liotta investigated graph drawings where edge crossings are only permitted at $90^\circ$ angles~\cite{DIDIMO20115156} (see Figure~\ref{fig:example_rac} for an illustration).
Today, these \emph{right-angle crossing} (or \emph{RAC}) drawings are among the best known and most widely studied beyond-planar drawing styles~\cite{DidimoLM19,Didimo20}, with the bulk of the research to date focusing on understanding necessary and sufficient conditions for the existence of
such drawings as well as the space they require~\cite{AngeliniCDFBKS11,ArgyriouBS12,GiacomoDEL14,GiacomoDGLR15,BekosDLMM17,AngeliniBFK20,AngeliniBKKP22,Forster020}.
A prominent theme in the context of RAC drawings concerns the number of times edges are allowed to be bent: it has been shown that every graph admits a RAC drawing if each edge can be bent $3$ times~\cite{DIDIMO20115156}, and past works have considered straight-line RAC drawings as well as RAC drawings where the number of bends per edge is limited to $1$ or $2$.

\begin{figure}
\vspace{-0.5cm}
    \centering
    \begin{subfigure}[t]{.4\linewidth}
        \centering
        \includegraphics{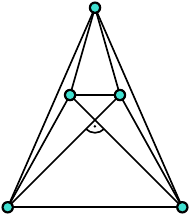}
        \caption{RAC drawing of $K_5$}
    \end{subfigure}
    \begin{subfigure}[t]{.59\linewidth}
        \centering
        \includegraphics{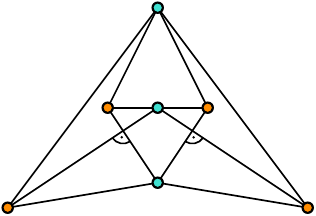}
        \caption{RAC drawing of $K_{3,3}$}
    \end{subfigure}
  \caption{Examples of RAC drawings.}
  \label{fig:example_rac}
\end{figure}

And yet, in spite of the considerable body of work concentrating on combinatorial and topological properties of such drawings, so far almost nothing is known about the complexity of computing a RAC drawing of a given graph. 
Indeed, while the problem of determining whether a graph admits a straight-line RAC drawing is \NP-hard~\cite{ArgyriouBS12} and was recently shown to be $\exists \mathbb{R}$-complete~\cite{Schaefer21}, there is a surprising lack of known algorithms that can compute such drawings for special classes of graphs or, more generally, parameterized algorithms that exploit quantifiable properties of the input graph to guarantee the tractability of computing RAC drawings (either without or with limited bends). 
This gap in our understanding starkly contrasts the situation for so-called $1$-planar drawings---another prominent beyond-planar drawing style for which a number of fixed-parameter algorithms are known~\cite{BannisterCE18,EibenGHKN20,EibenGHKN20b}---as well as recent advances mapping the boundaries of tractability for other graph drawing problems~\cite{HlinenyS19,BhoreGMN20,BhoreGMN22}.

\smallskip
\noindent
\textbf{Contribution.}\quad
We initiate an investigation of the parameterized complexity of determining whether a graph $G$ admits a RAC drawing. 
Given the well-motivated focus of previous works on limiting the amount of bends in such drawings, an obvious first choice for a parameterization would be to consider an upper bound $b$ on the total number of bends permitted in the drawing. 
However, on its own such a parameter cannot suffice to achieve fixed-parameter tractability in view of the \NP-hardness of the problem for $b=0$, i.e., for straight-line RAC drawings.

Hence, we turn towards identifying structural parameters of $G$ that guarantee fixed-parameter RAC drawing algorithms. 
While established decompositional parameters such as treewidth~\cite{RobertsonS84} and clique-width~\cite{CourcelleMR00} represent natural choices of parameterizations for purely combinatorial problems, the applicability of these parameters in solving graph drawing problems is complicated by the inherent difficulty of performing dynamic programming when the task is to obtain a drawing of the graph. 
This is why the parameters often used in this setting are non-decompositional, with the most notable examples being the \emph{vertex cover number} $\vc$ (i.e., the size of a minimum vertex cover) and the \emph{feedback edge number} $\fe$ (i.e., the edge deletion distance to acyclicity); further details are available in the overview of related work below.
As our main contributions, we provide two novel parameterized algorithms:

\begin{enumerate}
    \item a fixed-parameter algorithm for determining whether $G$ admits a RAC drawing with at most $b$ bends when parameterized by $\fe(G)$;
    \item a fixed-parameter algorithm for determining whether $G$ admits a RAC drawing with at most $b$ bends when parameterized by $\vc(G)+b$; 
\end{enumerate}

Both of the presented algorithms are constructive, meaning that they can also output a RAC drawing of the graph if one exists. 
The core underlying technique used in both proofs is that of \emph{kernelization}, which relies on defining reduction rules that can provably reduce the size of the instance until it is upper-bounded by a function of the parameter alone. 
While kernelization is a well-established and generic technique, its use here requires non-trivial insights into the structural properties of optimal solutions in order to carefully identify parts of the graph which can be simplified without impacting the final outcome.

We prove that both algorithms in fact hold for the more general case where each edge is marked with an upper bound on the number of bends it can support, allowing us to capture the previously studied $1$- and $2$-bend RAC drawings. 
Moreover, we show that the latter algorithm can be lifted to establish fixed-parameter tractability when parameterized by $b$ plus the \emph{neighborhood diversity} (i.e., the number of maximal modules) of $G$~\cite{Lampis10,Neighborhoodd,KnopKMT19}. 
In the concluding remarks, we also discuss possible extensions towards more general parameterizations and apparent obstacles on the way to such results.

\smallskip
\noindent
\textbf{Related Work.}\quad
Didimo, Eades and Liotta initiated the study of RAC drawings by analyzing the interplay between the number of bends per edge and the total number of edges~\cite{DIDIMO20115156}. 
Follow-up works also considered extensions and variants of the initial concept, such as upward RAC drawings~\cite{AngeliniCDFBKS11}, 2-layer RAC drawings~\cite{GiacomoDEL14,GiacomoDGLR15} and 1-planar RAC drawings~\cite{BekosDLMM17}. 
More recent works investigated the existence of RAC drawings for bounded-degree graphs~\cite{AngeliniBKKP22}, and RAC drawings with at most one bend per edge~\cite{AngeliniBFK20}. It is known that every graph admits a RAC drawing with at most three bends per edge~\cite{DIDIMO20115156}, and that determining whether a graph admits a RAC drawing with zero bends per edge is \NP-hard~\cite{ArgyriouBS12}.

The vertex cover number has been used as a structural graph parameter to tackle a range of difficult problems in graph drawing as well as other areas. 
Fixed-parameter algorithms for drawing problems based on the vertex cover number are known for, e.g., computing the obstacle number of a graph~\cite{BalkoCG00V022}, computing the stack and queue numbers of graphs~\cite{BhoreGMN20,BhoreGMN22}, computing the crossing number of a graph~\cite{HlinenyS19} and 1-planarity testing~\cite{BannisterCE18}. 
Similarly, the feedback edge number (sometimes called the \emph{cyclomatic number}) has been used to tackle problems which are not known to be tractable w.r.t.\ treewidth, including 1-planarity testing~\cite{BannisterCE18} and the \textsc{Edge Disjoint Paths} problem~\cite{GanianO21} (see also Table~1 in~\cite{GanianK22}).

These two parameterizations are incomparable: there are problems which remain \NP-hard on graphs of constant vertex cover number while being \FPT when parameterized by the feedback edge number (such as \textsc{Edge Disjoint Paths}~\cite{FleszarMS18,GanianO21}), and vice-versa. 
That being said, the existence of a fixed-parameter algorithm parameterized by the feedback edge number is open for a number of graph drawing problems that are known to be \FPT w.r.t.\ the vertex cover number; examples include computing the aforementioned stack, queue and obstacle numbers.

\section{Preliminaries}
\iflong We assume familiarity with standard concepts in graph theory~\cite{Diestel}. All graphs considered in this manuscript are assumed to be simple and undirected. \fi

\ifshort
We assume familiarity with standard concepts in graph theory~\cite{Diestel}, and with basic parameterized complexity theory~\cite{DowneyF13,Cygan}. 
All graphs considered in this manuscript are assumed to be simple and undirected.

The \emph{feedback edge number} of a graph $G$, denoted by $\fe(G)$, is the size of a minimum edge set $F$ such that $G-F$ is acyclic; it is well-known that such a set $F$ (and hence also the feedback edge number) can be computed in linear time.
The \emph{vertex cover number} of $G$, denoted $\vc(G)$, is the size of a minimum vertex cover of $G$, i.e., of a minimum set $X$ such that $G-X$ is edgeless. 
Such a minimum set $X$ can be computed in time $\bigO(1.2738^{|X|}+ |X|\cdot |V(G)|)$~\cite{ChenKX10}, and a vertex cover of size at most $2|X|$ can be computed in linear time by a trivial approximation algorithm. 
The third structural parameter considered here is the \emph{neighborhood diversity} $\nd(G)$ of $G$, which is the minimum size of a partition $\mathcal{P}$ of $V(G)$ such that for each pair $a,b$ in the same part of $\mathcal{P}$, it holds that $N(a)\setminus \{b\}=N(b)\setminus \{a\}$ where $N(a)$ and $N(b)$ are the open neighborhoods of $a$ and $b$, respectively. 
It is known that each part in such a partition $\mathcal{P}$ must be a clique or an independent set, and such a minimum partition can be computed in polynomial time~\cite{Lampis10}.
\fi

\smallskip
\noindent \textbf{RAC Drawings.}\quad
Given a graph $G = (V,E)$ on $n$ vertices with $m$ edges, a \emph{drawing} of $G$ is a mapping $\delta$ that takes vertices $V$ to points in the Euclidean plane $\R^2$, and assigns to every edge $e = uv \in E$ the image of a simple plane curve $[0,1] \rightarrow \R^2$ connecting the points $\delta(u), \delta(v)$ corresponding to $u$ and~$v$. 
We require that $\delta$ is injective on $V$, and furthermore that for all vertices $v$ and edges $e$ not incident to $v$, the point $\delta(v)$ is not contained in $\mathrm{int}(\delta(e))$, where $\mathrm{int}(\delta(e))$ is the image of $(0,1)$ under $\delta$.

A \emph{polyline drawing} of $G$ is a drawing such that for each edge $e\in E$, $\delta(e)$ can be written as a union $\delta(e) = \lambda^e_1 \cup \cdots \cup \lambda^e_t$ of closed straight-line segments $\lambda^e_1,\ldots,\lambda^e_t$ 
\iflong
such that:
\begin{itemize}[itemsep=0pt, topsep=0pt]
    \item for each $1\leq i\leq t-1$, the segments $\lambda^e_i$ and $\lambda^e_{i+1}$ intersect in precisely one of their shared end-points and moreover close an angle different than 180$^\circ$, and
    \item every other pair of segments is disjoint.
\end{itemize}
\fi

\ifshort such that (1) for each $1\leq i\leq t-1$, the segments $\lambda^e_i$ and $\lambda^e_{i+1}$ intersect in precisely one of their shared end-points and moreover close an angle different than 180$^\circ$, and (2) every other pair of segments is disjoint. \fi
The shared intersection points between consecutive segments are called the \emph{bends} of $e$ in the drawing $\delta$.

For two edges $e$ and $f$, their set of \emph{crossings} in the drawing $\delta$ is the set $\mathrm{int}(\delta(e)) \cap \mathrm{int}(\delta(f))$. 
We will assume without loss of generality that any drawing $\delta$ of $G$ has a finite number of crossings.

The central type of drawing studied in this paper are those that allow only \emph{right-angle crossings} between edge drawings (so-called \emph{RAC drawings}): 
We say that the edges $e,f \in E$ have a \emph{right-angle crossing} in a polyline drawing $\delta$ of $G$ if the crossing lies in the relative interiors of the respective line segments defining $\delta(e)$ and $\delta(f)$, and most crucially, the intersecting line segments of $\delta(e)$ and $\delta(f)$ are orthogonal to each other (i.e., they meet at a right angle).
Let $\delta$ be a polyline-drawing of a graph, $\beta: E \mapsto \{0,1,2,3\}$ a mapping, and $b \in \mathbb{N}$ a number. 
If every crossing of $\delta$ is a right-angle crossing, the number of bends counted over \emph{all} edges is at most $b$, and every edge itself has at most $\beta(e)$ bends, $\delta$ is called a \emph{$b$-bend $\beta$-restricted RAC drawing} of $G$.
We note that

\begin{itemize}[itemsep=0pt,topsep=0pt]
    \item  $0$-bend RAC drawings are straight-line RAC drawings (for any choice of $\beta$), 
    \item $m$ and $2m$-bend drawings with $\beta(e) = 1$ or $\beta(e) = 2$ for each edge $e$ gives the usual notion of $1$-bend and $2$-bend RAC drawings, respectively, and
    \item similarly, $3m$-bend drawings with $\beta(e)=3$ for each edge $e$ gives rise to the notion of $3$-bend RAC drawings, which exist for every graph~\cite{DIDIMO20115156}.
\end{itemize}

Based on the above, we can now formally define our problem of interest:

\begin{center}
    \begin{tcolorbox}[width=0.88\linewidth]
        \textsc{Bend-Restricted RAC Drawing (\BRAC)}\\
        \textbf{Input:} A graph $G$, an integer $b\geq 0$, and an edge-labelling $\beta: E \mapsto \{0,1,2,3\}$.\\
        \textbf{Question:} Does $G$ admit a $b$-bend $\beta$-restricted RAC drawing?
    \end{tcolorbox}
\end{center}

It has been shown that \textsc{$b$-bend $\beta$-restricted RAC Drawing} is $\exists \mathbb{R}$-complete~\cite{Schaefer21,Bieker20} even when restricted to the case where $b=0$. 
Without loss of generality, we will assume that the input graph $G$ is connected. 
We remark that while \BRAC\ is defined as a decision problem, every algorithm provided in this paper is constructive and can output a drawing as a witness for a \yesInstance{}.

\iflong
\smallskip
\noindent \textbf{Parameterized Algorithms.}\quad
We will not need a lot of the machinery of parameterized algorithms to state our results.
However, as it will turn out, our tractability results all come under the guise of so-called \emph{kernelization}, which requires some context.

A \emph{parameterized problem} is an ordinary decision problem, where each instance $I$ is additionally endowed with a \emph{parameter} $k$.
Given such a parameterized problem $\Pi$, we then say that a problem is \emph{fixed-parameter tractable} (\FPT) if there is an algorithm that, upon the input of an instance $(I,k)$ of $\Pi$, decides whether or not $(I,k)$ is a \yesInstance{} in time $f(k) \cdot n^{\bigO(1)}$, where $f$ is any computable function, and $n = |I|$ is the encoding length $|I|$ of the (parameter-free) instance $I$. 
This should be contrasted with parameterized problems that require time, say, $n^k$ to solve, which are not fixed-parameter tractable.

For instance, we may ask if a graph has a vertex-cover of size at most $k$, and declare $k$ the parameter of the instance. 
In this case, the problem is solvable in time $2^k \cdot n^{\bigO(1)}$, and hence \FPT; in contrast, asking for a dominating set of size $k$ (under some complexity assumptions) requires time $n^k$ for every $k$. 
Closer to the problems treated in this paper are structural parameterizations in the following sense: 
Suppose we are given a graph $G$ and a number $k$ such that $G$ has a vertex-cover of size at most $k$. 
Can we leverage this information to solve some (other) graph problem at hand? In this case, we say that we parameterize the problem \emph{by the vertex cover number}.

When using such parameterizations in our results, we will crucially rely on the following notion: 
A \emph{kernelization} (or kernel, for short) of $\Pi$ is a polynomial-time algorithm (in $n$, and we may assume $k \leq n$ holds) that takes an instance $(I,k)$ as input, and produces as output another instance $(I',k')$ with the following properties: there is some computable function $g$ such that both $k'$ and $|I'|$ are bounded from above by $g(k)$, and $(I,k)$ is a \yesInstance{} of $\Pi$ if and only if $(I',k')$ is.

That is, a kernelization algorithms preprocesses instances of arbitrary size into instances that are ``parameter-sized,'' and in particular (assuming $\Pi$ was decidable), this implies an algorithm running in time $n^{\bigO(1)} + h(g(k))$ for some function $h$ (where $h(g(k))$ is the running time of any algorithm solving instances of $\Pi$ of size $g(k)$). 
This means in particular that $\Pi$ is fixed-parameter tractable (and, as a standard result in parameterized algorithms, the converse of this claim holds as well). 
We refer to the standard textbooks~\cite{DowneyF13,Cygan} for a general treatment of parameterized algorithms.

The \emph{feedback edge number} of a graph $G$, denoted $\fe(G)$, is the size of a minimum edge set $F$ such that $G-F$ is acyclic. 
It is well-known that such a set $F$ (and hence also the feedback edge number) can be computed in linear time, since $G-F$ is a spanning tree of $F$. 
The \emph{vertex cover number} of $G$, denoted $\vc(G)$, is the size of a minimum vertex cover of $G$, i.e., of a minimum set $X$ such that $G-X$ is edgeless. 
Such a minimum set $X$ can be computed in time $\bigO(1.2738^{|X|}+ |X|\cdot |V(G)|)$~\cite{ChenKX10}, and a vertex cover of size at most $2|X|$ can be computed in linear time by a trivial approximation algorithm. 
The third structural parameter considered here is the \emph{neighborhood diversity} $\nd(G)$ of $G$, which is the minimum size of a partition $\mathcal{P}$ of $V(G)$ such that for each $a,b$ in the same part of $\mathcal{P}$ it holds that $N(a)\setminus \{b\}=N(b)\setminus \{a\}$. 
It is well known that each part in such a partition $\mathcal{P}$ must be either a clique or an independent set, and such a minimum partition can be computed in polynomial time~\cite{Lampis10}.
\fi

\section{An Explicit Algorithm for \BRAC}

As already pointed out above, our results for fixed-parameter tractability come as kernels.
While there is a generic formal equivalence between the existence of a kernel and a decidable problem being fixed-parameter tractable, this doesn't by itself yield explicit bounds on the running time of the algorithm that results from this generic strategy. 
In order to derive concrete upper bounds on the running time of our algorithms, we provide an algorithm that solves \emph{$b$-bend $\beta$-restricted RAC drawing} with a specific running time bound. 
We do so via a combination of branching and an encoding in the existential theory of the reals.

\iflong \begin{theorem} \fi
\ifshort \ifreference \begin{theorem}[$\clubsuit$] \else \begin{theorem} \fi \fi
\label{theorem:explicit_runtime}
    An instance $(G,b,\beta)$ of \BRAC\ can be solved in time $m^{\bigO(m^2)}$, where $m$ is the number of edges of $G$.
\end{theorem}

\ifshort
    \begin{proof}[Sketch]
    We begin with a branching step in which we exhaustively consider all possible allocations of the bends to edges. 
    In each branch, we alter the graph $G$ by subdividing each edge precisely the number of times it is assumed to be bent in that branch. 
    At this point, it remains to decide whether this new graph $G'$ admits a straight-line RAC drawing. 
    To do this, one can construct a sentence in the existential theory of the reals that is true if and only if $G'$ admits such a drawing~\cite{Bieker20}. 
    To conclude the proof, we use the known fact that an existential sentence over the reals in $N$ variables over $M$ polynomials of maximal degree $D$ can be decided in time $(M \cdot D)^{\bigO(N)}$ (see, e.g.,~\cite[Theorem 13.13]{basu2006algorithms}).
    \qed    \end{proof}
\fi

\iflong 
\begin{proof}
    Observe that, without loss of generality, we may assume that $b\leq 3m$. 
    We begin by a branching step in which we exhaustively consider all possible allocations of the bends to edges, resulting in a total number of at most $4^m$ branches (some of which will be discarded due to exceeding the bound $b$ or violating $\beta$). 
    In each branch, we alter the graph $G$ by subdividing each edge precisely the number of times it is assumed to be bent in that branch. 
    At this point, it remains to decide whether this new graph $G'$ admits a straight-line RAC drawing, where $G'$ has $\bigO(m)$ edges and vertices, and we denote these $m'$ and $n'$, respectively.

    To do this, one can construct a sentence in the existential theory of the reals that is true if and only if $G'$ admits such a drawing. 
    The variables of the sentence will consist of $n'$ variable pairs $(x_{v_1},y_{v_1}),\ldots,(x_{v_{n'}},y_{v_{n'}})$, encoding the coordinates of the drawing of the vertices in $\R^2$.
    Furthermore, for every pair of edges with endpoints $u,v$ and $u',v'$, we can formulate a condition $\sigma(u,v,u',v') \Rightarrow \tau(u,v,u',v')$, where $\sigma$ is a polynomial condition in $x_u,x_v,x_u',x_v'$ encoding whether the straight-line segments corresponding to $uv$ and $u'v'$ intersect, and $\tau$ is a polynomial condition in $x_u,x_v,x_u',x_v'$ encoding whether these straight-line segments are perpendicular. 
    Indeed, the former requires an addition of another $m'^2$ auxiliary variables in the worst case, but both conditions can be expressed by polynomials of degree two. 
    This encoding is described in full detail by Bieker~\cite{Bieker20}.

    To conclude the proof, we note that  an existential sentence over the reals in $N$ variables over $M$ polynomials of maximal degree $D$ can be decided in time $(M \cdot D)^{\bigO(N)}$ (see, e.g.,~\cite[Theorem 13.13]{basu2006algorithms}).
    Note that, within essentially the same running time bound, one can also construct a representation of a solution for this system \cite[Theorem 13.11]{basu2006algorithms}.
\qed    \end{proof} 
\fi

\section{A Fixed-Parameter Algorithm via $\fe(G)$}
\label{sec:fe}

We begin our investigation by establishing a kernel for \textsc{Bend-Restricted RAC Drawing} when parameterized by the feedback edge number.
Our kernel is based on the exhaustive application of two reduction rules.

Let us assume we are given an instance $(G,b,\beta)$ of \BRAC\ and that we have already computed a minimum feedback edge set $F$ of $G$ in linear time.
The first reduction rule is trivial: we simply observe that vertices of degree one can always be safely removed since they never hinder the existence of a RAC drawing.

\iflong \begin{observation} \fi
\ifshort \ifreference \begin{observation}[$\clubsuit$] \else \begin{observation} \fi \fi
    \label{obs:fenleaves}
    Let $v \in V(G)$ be a vertex with degree one. 
    $G-\{v\}$ admits a $b$-bend $\beta$-restricted RAC drawing if and only if $G$ does as well.
\end{observation}

\iflong 
\begin{proof}
    Clearly, if $G$ admits a $b$-bend $\beta$-restricted RAC drawing, then $G-\{v\}$ does as well (one may simply remove $v$ and its incident edge from the drawing). 
    On the other hand, if $G-\{v\}$ admits a $b$-bend $\beta$-restricted RAC drawing then we can extend this drawing to one for $G$ by placing $v$ sufficiently close to its only neighbor in a way which does not induce any additional crossings.
\qed    \end{proof} 
\fi

\begin{figure}
    \centering
    \begin{subfigure}[t]{.49\linewidth}
        \centering
        \includegraphics{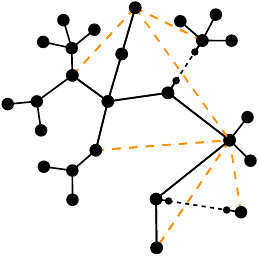}
        \caption{Before removing degree-one vertices.}
    \end{subfigure}
    \begin{subfigure}[t]{.49\linewidth}
        \centering
        \includegraphics{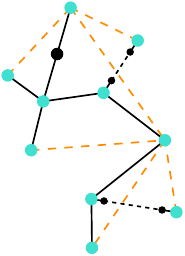}
        \caption{After removing degree-one vertices.}
    \end{subfigure}
    \caption{Reduction rule one: Degree-one vertices can be pruned. 
    Orange lines represent feedback edges, dashed black lines represent long paths and special vertices are marked in turquoise.}
    \label{fig:FE_Remove_Subtrees}
\end{figure}

Iteratively applying the reduction rule provided by Observation~\ref{obs:fenleaves} results in a graph of the form $G' = (V', E' \cup F)$, where $T := (V',E')$ is a tree with at most $2 \cdot \fe(G)$ leaves and where each leaf of $T$ is incident to at least one edge in $F$. 
We mark a vertex in $T$ as \emph{special} if it is an endpoint of an edge in $F$ or if it has degree at least $3$ in $T$ (see Figure~\ref{fig:FE_Remove_Subtrees} for an illustration). 
\ifshort The total number of special vertices can be upper-bounded by $4\cdot \fe(G)$. \fi
\iflong Note that the total number of special vertices is upper-bounded by $4\cdot \fe(G)$: the total number of endpoints of edges in $F$ is bounded by $2\cdot \fe(G)$, and since this also upper-bounds the number of leaves this implies that there can be at most $2\cdot \fe(G)$ vertices of degree at least $3$ in $T$. \fi

In order to define the crucial second reduction rule, we will partition the edges of $T$ into edge-disjoint paths such that each special vertex can only appear as an endpoint in such paths.

\begin{definition}
    We define the \emph{path partition} of $T$ in $G^\prime$ as the unique partition $P_1 \, \dot \cup \cdots \dot \cup\, P_\ell = E'$ such that
    all $P_i$ are pairwise edge-disjoint paths in $T$ whose endpoints are both special vertices, but with no special vertices in their interior.  
    We call $\ell$ the size of the path partition.
\end{definition}

An illustration is provided in Figure~\ref{fig:path_partition}.
Given the established bound on the number of special vertices, the size of the path partition is bounded by $4\cdot \fe(G)$.

\begin{figure}
    \centering
    \includegraphics{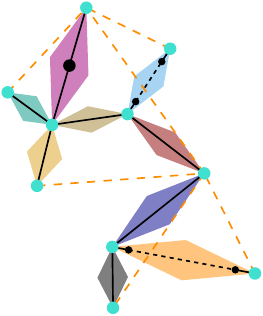}
    \caption{Path partition of $T$ with feedback edges in orange.}
    \label{fig:path_partition}
\end{figure}

At this point, let us assume that we have a path partition $P_1 \, \dot \cup \, P_2 \, \dot \cup \cdots \dot \cup\, P_{\ell}$ of $T$ in $G^\prime$, where we index the paths in increasing order of length.
Our next task is to divide these paths into short and long paths by identifying whether there exists a large gap in the lengths of these paths.

\begin{definition}
    Define $p_i := |P_i|$ for $i = 1, \dots, \ell$, and moreover define $P_0:=F$ and $p_0:=|F|$.
    Let $i_0$ be the minimal $i = 1, \dots, \ell$ such that $p_i > 9 \ell \cdot p_{i-1}$, if one such $i$ exists, otherwise we set $i_0 := \ell$.
    We call all paths $P_i$ with $1\leq i \leq i_0$ \emph{short} and all other paths \emph{long}.
    Then we define the subgraph $G_{\mathrm{short}}$ as the edge-induced subgraph of $\bigcup_{i=0}^{i_0} P_i$ of $G^\prime$ (i.e., $G_{\mathrm{short}}$ arises by removing all long paths from $G^\prime$).
\end{definition}

Our aim is now to argue that if $\delta_{\mathrm{short}}$ is a RAC drawing of $G_{\mathrm{short}}$, then we can always extend $\delta_{\mathrm{short}}$ to a RAC drawing of $G^\prime$. 
Without loss of generality we assume that all vertices in $V(G^\prime)$ have already been drawn in $\delta_{\mathrm{short}}$. 
First we create an intermediate drawing $\delta^\prime$ of $G^\prime$, which will in general not be a RAC drawing.
We define $\delta^\prime$ as an extension of $\delta_{\mathrm{short}}$, where each long path $P$ with endpoints $s$ and $t$ is represented as a simple straight-line segment from $\delta_{\mathrm{short}}(s)$ to $\delta_{\mathrm{short}}(t)$ with all interior vertices distributed arbitrarily along that line segment.
Doing this will in general violate the RAC property of $\delta^\prime$, hence in the next step we need to alter this straight-line segment in order to ensure that the drawing of $P$ crosses only at right angles.
For this we observe that any vertex on $P$ can be moved to effectively act as a bend in a polyline drawing of $P$.
We show that these ``additional bends'' can be used to turn all crossings into right-angle crossings.

\iflong \begin{lemma} \fi
\ifshort \ifreference \begin{lemma}[$\clubsuit$] \else \begin{lemma} \fi \fi
    \label{lemma:polyline_construction}
    Let $P$ be a long path with endpoints $s$ and $t$ and consider its straight-line representation $L$ in $\delta^\prime$.
    Assume $L$ intersects $k$ straight-line segments in $\delta^\prime$.    
    Then, there exists a polyline segment $L^\star$ from $\delta^\prime(s)$ to $\delta^\prime(t)$ with at most $3k$ bends that intersects precisely the line segments intersected by $L$, where each such segment is crossed precisely once and at a right angle.
\end{lemma}

\ifshort
\begin{proof}[Sketch]
    The proof considers several ways in which $L$ can intersect a line segment in $\delta^\prime$ (such as touching an endpoint of the segment, or crossing through its interior, or containing the line segment). 
    In each of these cases, we show that it is possible to add a small number of bends in the immediate vicinity of the respective segment to create a right-angle crossing with the interior of the segment. 
    An intuitive illustration of the main case is provided in Figure~\ref{fig:new_rac_in_epsilon_2}.
\qed    \end{proof}

\begin{figure}
    \centering
    \includegraphics{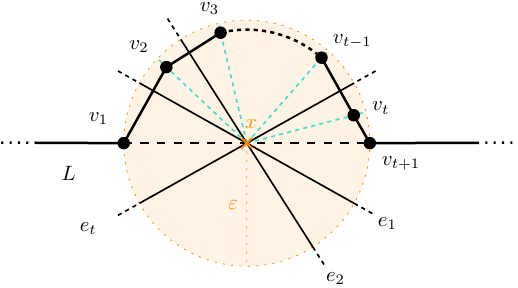}
    \caption{An illustration of how bends can be used to deal with the situation where multiple straight-line segments in $\delta^\prime$ intersect at a single point.}
    \label{fig:new_rac_in_epsilon_2}
\end{figure}
\fi

\iflong 
\begin{proof}
    For the purposes of this proof, it will be useful to treat each bend as an auxiliary vertex in $\delta^\prime$ and treat straight-line segments as edges.
    Let $e \in E(G_{\mathrm{short}})$ be an edge such that $\delta^\prime(e)$ is crossed by $L$ at the point $x$. 
    We now distinguish the following cases of how $L$ intersects the straight-line segments in $\delta'$, and deal with each case separately:

    \begin{enumerate}
        \item $x \in \mathrm{int}(\delta^\prime(e))$ and no other edge crosses through $x$.
        \item $\exists \, f \in E(G_{\mathrm{short}}): f \neq e \land x \in \mathrm{int}(\delta^\prime(e)) \cap \mathrm{int}(\delta^\prime(f))$.
        \item $\exists \, v \in V(G_{\mathrm{short}}): x = \delta^\prime(v)$.
    \end{enumerate}

    We show how to deal with each of these three cases below:
    
    \begin{enumerate}
        \item Let $\varepsilon > 0$ such that $B_{\varepsilon}(x) = \{y \in \mathbb{R}^2: |x - y| < \varepsilon\}$ contains no vertices and intersects no other edges outside of $L$ apart from $e$.
        We convert the intersection at $x$ to a right-angle crossing by introducing three bends on the boundary $\partial B_{\varepsilon}(x)$ as illustrated in Figure \ref{fig:new_rac_in_epsilon_1}:
        Put two vertices $v_1$, $v_3$ on the intersection of $L$ with $\partial B_{\varepsilon}(x)$ to maintain the position of $L$ outside of $B_{\varepsilon}(x)$. 
        Construct the middle vertex $v_2$ by taking the intersection of $\partial B_{\varepsilon}(x)$ with the normal line to $e$ going through $v_3$.
        Therefore we obtain our new polyline $L^\star$ by joining the parts of $L$ outside the $\varepsilon$-neighborhood with the polyline connecting the three vertices on the boundary of the $\varepsilon$-neighborhood.

        \begin{figure}[hbt!]
            \centering
            \includegraphics{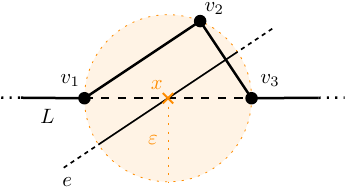}
            \caption{Dealing with a single crossing.}
            \label{fig:new_rac_in_epsilon_1}
        \end{figure}
        
        Since we chose $\varepsilon$ such that there are no further intersections within this $\varepsilon$-neighborhood and the polyline remains unchanged outside of this neighborhood, we are guaranteed to not introduce any new crossings nor alter existing ones.

        \item Generalizing the first case, we now consider an intersection point $x$ with multiple edges crossing it.
        We denote with $C = \{e_1,\dots,e_t\}$ the set of edges intersecting at $x$.
        Let $\varepsilon > 0$ such that $B_{\varepsilon}(x)$ contains no vertices and intersects no other edges apart from the edges in $C$.
        Again, two entry and exit nodes $v_1$, $v_{t+1}$ are added on the boundary of $B_{\varepsilon}(x)$ to preserve the original position of $L$ outside of $B_{\varepsilon}(x)$. 
        Assume the edges in $C$ are ordered clockwise with $e_1$ being the closest edge to $v_1$.
        Now we iteratively construct the next vertex $v_{i+1}$ by taking the intersection of the angular bisector between $e_i$ and $e_{i+1}$ with the normal line to $e_i$ going through $v_i$.
        If this point happens to lie outside of $B_{\varepsilon}(x)$ we take the intersection of the normal line with $\partial B_{\varepsilon}(x)$ instead. We refer to Figure \ref{fig:new_rac_in_epsilon_2} for an illustration.

        \begin{figure}[hbt!]
            \centering          
            \includegraphics{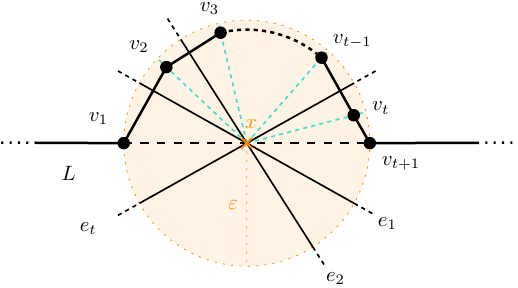}
            \caption{Crossing multiple edges at once.}
            \label{fig:new_rac_in_epsilon_2}
        \end{figure}

        In total, we need at most $t+1 \leq 3t$ vertices.

        \item In the third case, the straight-line segment $L$ intersects with a vertex $v$ at $x = \delta^\prime(v)$. 
        If $L \cap \delta^\prime(e) = \delta^\prime(v)$, i.e. $L$ does not run parallel to $\delta^\prime(e)$, we can simply take care of this equivalently as if $L$ would cross $\delta^\prime(e)$ in the interior.
        
        If otherwise $L \cap \delta^\prime(e) \supsetneq \delta^\prime(v)$, we observe that $L$ must necessarily contain both endpoints of $\delta^\prime(e)$, since the endpoints $\delta^\prime(s)$ and $\delta^\prime(t)$ of $L$ cannot lie in the interior of $\delta^\prime(e)$. 
        At the first of these endpoints encountered by $L$, say $w$, we again proceed analogously to the second case (i.e., as if $w$ was a crossing point), however instead of exiting the circle surrounding $w$ directly opposite to the entry point we exit at a distance of $\varepsilon'$ from there (for a sufficiently small $\varepsilon'$) and from there draw $L$ in parallel to $e$. 
    \end{enumerate}
    Since these three cases are exhaustive, the lemma follows.
\qed    \end{proof}
\fi

\iflong \begin{lemma} \fi
\ifshort \ifreference \begin{lemma}[$\clubsuit$] \else \begin{lemma} \fi \fi
    \label{lemma:number_of_crossings}
    Each long path intersects at most $3\ell \cdot p_{i_0}$ straight-line edge segments in $\delta^\prime$.
\end{lemma}

\iflong 
\begin{proof}
    Since each long path is represented as a straight-line segment in $\delta^\prime$, it can cross every other long path at most once.
    Since every edge $e \in E(G_{\mathrm{short}})$ can be bent at most $\beta(e) \leq 3$ times, $L$ can cross $e$ at most three times.
    Let $t$ be the number of long paths, then $L$ intersects no more than
    \[
        t + 3\sum_{i=0}^{i_0}p_i \leq t + 3(\ell - t)\cdot p_{i_0} \leq 3\ell \cdot p_{i_0}
    \] 
    straight-line edge segments.
\qed    \end{proof} 
\fi

\iflong \begin{theorem} \fi
\ifshort \ifreference \begin{theorem}[$\clubsuit$] \else \begin{theorem} \fi \fi
    \label{thm:fenkernel}
    \textsc{$b$-bend $\beta$-restricted RAC Drawing} admits a kernel of size at most $(36 \cdot \fe(G))^{4 \cdot \fe(G)}$.
    The kernel can be constructed in linear time.
\end{theorem}

\ifshort
\begin{proof}[Sketch]
    Consider an input $(G,b,\beta)$ with a feedback edge set $F$ of $G$. 
    We begin by exhaustively removing all vertices of degree one, as per Observation~\ref{obs:fenleaves}. 
    Next, we construct a path partition $\mathcal{P} = (P_1,\dots,P_\ell)$ of the tree $T = (V',E')$ in $G^\prime$ of size at most $4 \cdot \fe(G)$. 
    Then we split the paths in $\mathcal{P}$ into short and long paths, and we construct the subgraph $G_{\mathrm{short}}$ of $G'$ by removing all long paths from $G'$.
    To conclude the proof, it suffices to establish that $G_{\mathrm{short}}$ is a kernel, and in particular that $(G,b,\beta)$ is a \yesInstance\ of \BRAC\ if and only if $(G_{\mathrm{short}},b,\beta')$ is as well (where $\beta'$ is a restriction of $\beta$ to the edges in $G_{\mathrm{short}}$). 
    This can be shown by applying Lemmas~\ref{lemma:polyline_construction} and~\ref{lemma:number_of_crossings}.
\qed    \end{proof}
\fi

\iflong 
\begin{proof}
    Consider an input $(G = (V,E),b,\beta)$ with feedback edge set $F$. 
    In the first step according to Observation~\ref{obs:fenleaves} we iteratively prune all vertices of degree one and obtain the reduced graph $G' = (V', E' \cup F)$. 
    Next, we construct a path partition $\mathcal{P} = (P_1,\dots,P_\ell)$ of the tree $T = (V',E')$ of size at most $4 \cdot \fe(G)$. Then we split the paths in $\mathcal{P}$ into short and long paths. 
    We define the subgraph $G_{\mathrm{short}}$ as the graph obtained by removing all long paths from $G'$ and show that it is a kernel.
    
    \begin{figure}[hbt!]
        \centering
        \begin{subfigure}[t]{.49\linewidth}
            \centering
            \includegraphics{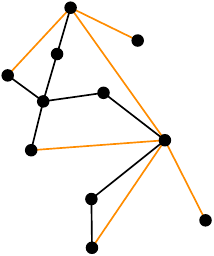}
            \caption{RAC drawing $\delta_{\mathrm{short}}$ of $G_{\mathrm{short}}$.}
            \label{fig:delta_short}
        \end{subfigure}
        \begin{subfigure}[t]{.49\linewidth}
            \centering
            \includegraphics{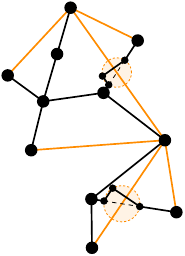}
            \caption{RAC drawing $\delta$ of $G^\prime$}
            \label{fig:delta}
        \end{subfigure}
        \caption{Extending the RAC drawing.}
    \end{figure}

    For that consider a RAC drawing $\delta_{\mathrm{short}}$ of $G_{\mathrm{short}}$. 
    We show that we can construct a RAC drawing $\delta$ of $G^\prime$ that extends $\delta_{\mathrm{short}}$ (see Figure~\ref{fig:delta}).
    To achieve this, we first define the intermediate drawing $\delta^\prime$, which extends $\delta_{\mathrm{short}}$ by simply drawing all long paths as straight-line segments.
    To obtain $\delta$ from $\delta^\prime$ we now iteratively replace the straight-line representations by the polyline constructions described in Lemma~\ref{lemma:polyline_construction}.
    Let $P_{i_0+1}$ be the first long path with straight-line drawing $L_{i_0 +1}$. 
    According to Lemma~\ref{lemma:number_of_crossings} it is involved in at most $3 \ell \cdot p_{i_0}$ crossings in $\delta^\prime$.
    By the definition of long paths we have $|P_{i_0+1}| > 9\ell \cdot p_{i_0}$ and thus we have enough vertices per crossing to construct $L_{i_0+1}^*$ via Lemma~\ref{lemma:polyline_construction}. 
    Define $\delta_{i_0 + 1}$ by replacing $L_{i_0+1}$ with $L_{i_0+1}^\star$ in $\delta^\prime$.
    Since $L_{i_0+1}^\star$ does not introduce any new crossings, we can repeat this process for all further long paths until we obtain $\delta := \delta_{\ell}$ as our final RAC drawing of $T$.
    Thus, we obtain $G_{\mathrm{short}}$ as a kernel of our original instance with $b$ and $\beta$ unchanged, since we did not use any additional bends in the construction of $\delta$.

    Next, we show that $G_{\mathrm{short}}$ can be constructed in linear time.
    We already observed that a Feedback Edge Set can be constructed in linear time.
    Pruning vertices of degree one can be done in linear time as well, while the task of finding a path partition of $T$ can be achieved by depth-first search in linear time as well.
    Finally, the size of $G_{\mathrm{short}}$ can be bounded by

    \begin{align*}
        \sum_{i=0}^{i_0} p_i 
        &\leq \sum_{i=0}^{\ell} p_0 (9\ell \cdot \fe(G))^i \\
        &\leq 2 \cdot \fe(G) (9\ell \cdot \fe(G))^{\ell} \\
        &\leq 2 (36 \cdot \fe(G))^{4 \cdot \fe(G)}. \qedhere
    \end{align*}
\end{proof}
\fi

Using \Cref{thm:fenkernel}, the runtime guarantee given by \Cref{theorem:explicit_runtime} and the fact that a feedback edge set of size $\fe(G)$ can be computed in linear time, we obtain:

\begin{corollary}
    \textsc{$b$-bend $\beta$-restricted RAC Drawing} is fixed-parameter tractable parameterized by $\fe(G)$, and in particular can be solved in time $2^{\fe(G)^{\bigO(\fe(G))}} + \bigO(|V(G)|)$.
\end{corollary}

\iflong
\begin{proof}
    After constructing the kernel in $\bigO(|V(G)|)$ time, we apply the generic branching algorithm to the kernel with a runtime of 
    \begin{align*}
        ((36 \cdot \fe(G))^{4 \cdot \fe(G)})^{\bigO\left((36 \cdot \fe(G))^{8 \cdot \fe(G)}\right)} &= (36 \cdot \fe(G))^{\bigO\left((36 \cdot \fe(G))^{8 \cdot \fe(G) + 1}\right)} \\
        &\leq 2^{(\log \fe(G)) \cdot \fe(G)^{\bigO(\fe(G))}} \\
        &= 2^{\fe(G)^{\bigO(\fe(G))},}
    \end{align*}
    which concludes the proof.
\qed    \end{proof}
\fi

\section{Fixed-Parameter Tractability via $\vc(G)$}

As in Section~\ref{sec:fe}, the core tool used to establish fixed-parameter tractability for this parameterization is a kernelization procedure, although the ideas and reduction rules used here are very different. 
Let us assume we are given an instance $(G,b,\beta)$ of \BRAC; as our first step, we compute a vertex cover $C$ of size $k\leq 2 \cdot \vc(G)$ using the standard approximation algorithm.

We now partition the vertices of our instance $G$ outside of the vertex cover $C$ into \emph{types}, as follows. 
Two vertices in $G\setminus C$ are of the same type if they have the same set of neighbors in $C$; observe that the property of ``being in the same type'' is an equivalence relation, and when convenient we also use the term \emph{type} to refer to the equivalence classes of this relation. 
To avoid any confusion, we explicitly remark that two vertices may have the same type even when their incident edges are assigned different values by $\beta$.
The number of types is upper-bounded by $2^{k}$. 

We distinguish types by the number of neighbors in $C$; an illustration is provided in \Cref{fig:vc_types_overview}. Let a \textit{member} of a type $T$ be defined as a vertex in $T$ as well as its incident edges. 
By an exhaustive application of the first reduction rule introduced in Section~\ref{sec:fe} (cf. Observation~\ref{obs:fenleaves}), we may assume that there is no type with less than $2$ neighbors in $C$.

\begin{figure}[hbt!]
    \centering
    \includegraphics{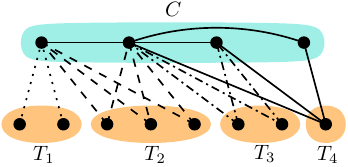}
    \caption{A graph split into its vertex cover $C$ (in turquoise) and its different types $T_1,\dots,T_4$ (in orange).}
    \label{fig:vc_types_overview}
\end{figure}

Turning to types with at least $3$ neighbors in $C$, we provide a bound on the size of each such type in a \yesInstance{} of $\BRAC$.

\begin{lemma}
    \label{lemma:bigger_types}
    If $(G,b,\beta)$ is a \yesInstance{} of \BRAC, then each type $T$ with $i \geq 3$ neighbors in $C$ has at most $\max(2,7-i)+b$ members.
\end{lemma}

\begin{proof}
    Didimo, Eades and Liotta showed that no complete bipartite graph $K_{c,d}$ with $c + d > 7$ and $\min(c,d) > 2$ admits a straight-line RAC drawing~\cite{Didimo}. 
    Hence, if vertices in $T$ have $3$ neighbors in $C$ then a $b$-bend $\beta$-restricted RAC drawing of $G$ can contain at most $4$ members of $T$ without bends; otherwise, the drawing of $5$ members of $T$ and their $3$ neighbors in $C$ would contradict the first sentence. 
    Similarly, if vertices in $T$ have at least $5$ neighbors in $C$ then a $b$-bend $\beta$-restricted RAC drawing of $G$ cannot contain $3$ members of $T$ without bends.
\qed    \end{proof} 

\Cref{lemma:bigger_types} implies that we can immediately reject instances containing types with more than $3$ neighbors whose cardinality is greater than $4+b$ (or, for the purposes of kernelization, one may replace these with trivial \noInstances{}).
Hence, it now remains to deal with types with precisely two neighbors in $C$.

\iflong
We say that two edges $uv$ and $uv'$ form a \textit{fan anchored} at $u$. 
It is easy to observe that if an edge $e$ crosses both $uv$ and $uv'$ in a $b$-bend $\beta$-restricted RAC drawing, then at least one of these three edges must have a bend~\cite{AngeliniCDFBKS11}. 
\fi

\iflong \begin{lemma} \fi
\ifshort \ifreference \begin{lemma}[$\clubsuit$] \else \begin{lemma} \fi \fi
    \label{lemma:crossing_edges}
    Consider a $b$-bend $\beta$-restricted RAC drawing $\delta$ of $G$, and let $T$ be a type containing vertices with precisely two neighbors in $C$. 
    Let $T'$ be the subset of $T$ containing all members of $T$ which do not have bends in $\delta$. 
    Then $T'$ contains at most four members involved in crossings with other members of $T'$ in~$\delta$.
\end{lemma}

\iflong \begin{proof} \fi
\ifshort \begin{proof}[Sketch] \fi
    Let $u$ and $v$ be the neighbors of $T'$ in the vertex cover $C$.
    Let us consider the vertices lying on one side of the half plane induced by the line $\overleftrightarrow{\delta(u)\delta(v)}$ going through $u$ and $v$. 
    According to Thales's theorem, every right-angle crossing formed by two edges originating in $u$ and $v$ respectively, has to lie on the semicircle with diameter $\overline{\delta(u)\delta(v)}$. 
    Suppose the edges $(u,x)$ and $(v,w)$ cross at a right angle.
    Then there cannot be another edge incident to $u$ which crosses the semicircle to the right of the first crossing
    \ifshort~(see Figure~\ref{fig:lemma_crossing_edges})\fi.

    \iflong
    Indeed, if there was such an edge $(u,y)$, then $(v,w)$ would cross the fan formed by $(u,x)$ and $(u,y)$ as shown in Figure~\ref{fig:lemma_crossing_edges}. 
    Analogously, there also cannot be an edge incident to $v$, which crosses the semicircle left to the first crossing.
    Hence, there can be at most one right-angle crossing between members of $T'$ below and above $\overleftrightarrow{\delta(u)\delta(v)}$, respectively.
    \fi
\qed    \end{proof} 

\begin{figure}
    \centering
    \includegraphics{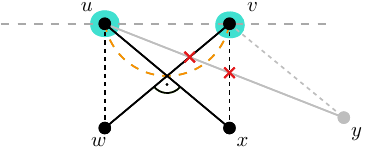}
    \caption{Illustration for \Cref{lemma:crossing_edges}.}
    \label{fig:lemma_crossing_edges}
\end{figure}

Next, we use the above statement to obtain a bound on the total number of crossings that such a type $T$ can be involved in. 
We do so by showing that the members of $T$ which themselves do not have bends are only involved in a bounded number of crossings.

\iflong \begin{lemma} \fi
\ifshort \ifreference \begin{lemma}[$\clubsuit$] \else \begin{lemma} \fi \fi
    \label{lemma:max_of_type_2}
    Consider a $b$-bend $\beta$-restricted RAC drawing $\delta$ of $G$, and let $T$ be a type containing vertices with precisely two neighbors in $C$. 
    Then at most $3k+6+b$ members of $T$ can be involved in a crossing in $\delta$.
\end{lemma}

\iflong \begin{proof} \fi
\ifshort \begin{proof}[Sketch] \fi
    Let $T'$ be the subset of $T$ containing all members of $T$ which do not have bends in $\delta$, and let $\gamma=|T|-|T'|$. 
    Further, let $T_0$ be the set of members of $T'$ which are pairwise crossing free, but which all cross at least some other edge in $\delta$. 
    $T_0$ forms a layering structure in $\delta$, as depicted in \Cref{fig:proof_lemma_2}. 
    Moreover, if $T_0$ contains two members that are incident to the same inner face in this layering structure and whose edges are drawn in parallel in $\delta$, we remove one of these members from $T_0$; observe that this may only reduce the size of $T_0$ by one. 
    Let $\alpha$ be the number of members that remain in $T_0$ at this point. 
    \ifshort To complete the proof, it suffices to use the fact that no straight-line segment may cross more than one edge in this layering structure in order to upper-bound $\alpha$ by $3k+1+b-\gamma$, and apply Lemma~\ref{lemma:crossing_edges} to bound the number of members in $T'$ involved in crossings. \fi
    
    \iflong
    At this point, an edge $e$ without any bends cannot cross more than one member of $T_0$, as no two edges on the same face in $T_0$ are parallel by assumption. 
    Without bends, this would imply that there must be a vertex in every layer, and  since each vertex can only be connected to other vertices in the same or in one of the two adjacent layers there can be at most $3k$ layers.
    Introducing bends, an edge outside $T'$ might cross one additional layer per bend; thus increasing the number of possibly crossed members to $3k+1+b$. Since $\gamma$ bends are already used by edges of $T$, we obtain $\alpha \leq 3k+1+b-\gamma$.

    Moving our attention to $T' \supseteq T_0$, the difference between the sizes of these two sets can be caused (1) by up to 4 members that are involved in crossings with other members of $T'$ following \Cref{lemma:crossing_edges} and (2) by one additional member for the single removed member with parallel edges from $T_0$, i.e. $|T'| \leq \alpha+5$. 
    Hence, at most $\alpha+5+\gamma = 3k+6+b$ members of $T$ can be involved in crossings in $\delta$.
    \fi
\qed    \end{proof}

\begin{figure}
    \centering
    \vspace{-0.8cm}
    \includegraphics{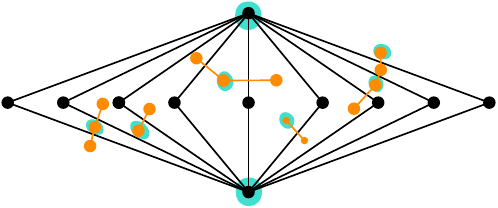}
    \caption{Illustration for the proof of \Cref{lemma:max_of_type_2} with vertices in the vertex cover marked in turquoise.}
    \vspace{-0.3cm}
    \label{fig:proof_lemma_2}
\end{figure}

In particular, Lemma~\ref{lemma:max_of_type_2} implies that in a $b$-bend $\beta$-restricted RAC drawing, every sufficiently large type $T$ with precisely $2$ neighbors in $C$ must contain a member that is not involved in any crossings. 
The next lemma highlights why this is useful in the context of our kernelization.

\iflong \begin{lemma} \fi
\ifshort \ifreference \begin{lemma}[$\clubsuit$] \else \begin{lemma} \fi \fi
    \label{lemma:type2_no_crossing}
    Let $T$ be a type with two neighbors in $C$ and assume that $G$ admits a $b$-bend $\beta$-restricted RAC drawing $\delta$. 
    If there is a member in $T$ whose edges are drawn without crossings in $\delta$, then the graph obtained from $G$ by adding a vertex $w'$ to $T$ admits a $b$-bend $\beta$-restricted RAC drawing as well.
\end{lemma}

\iflong 
\begin{proof}
    Let $u,v$ be the neighbors of $T$ in the vertex cover and let $w \in T$ be the member without crossings. 
    We can draw $w'$ infinitesimally close to $w$ such that the emerging layering triangles are drawn without crossings (see Figure \ref{fig:lemma_type2_no_crossing}).
    \begin{figure}
        \centering
        \includegraphics{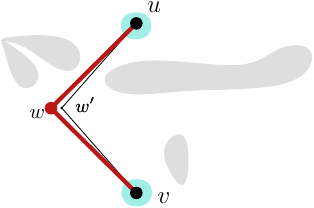}
        \caption{Illustration for the proof of \Cref{lemma:type2_no_crossing}.}
        \label{fig:lemma_type2_no_crossing}
    \end{figure}
\qed    \end{proof} 
\fi

At this point, we have all the ingredients for the main result of this section:

\begin{theorem}
    \label{theorem:vertex_cover}
    \textsc{$b$-bend $\beta$-restricted RAC Drawing} admits a kernel of size $\bigO(b \cdot 2^{k})$, where $k$ is the size of a provided vertex cover of the input graph.
\end{theorem}

\begin{proof}
    Consider an input $(G,b,\beta)$ and let $C$ be the provided vertex cover of $G$. 
    We apply the simple reduction rule of deleting vertices of degree $1$ from $G$, resulting in an instance where each type has either $2$ or at least $3$ neighbors in $C$. 
    For each type of the latter kind, we check if it contains more members than $\max(3,7-i)+b$; if yes, we reject (or, equivalently, replace the instance with a trivial constant-size \noInstance{}), and this is correct by Lemma~\ref{lemma:bigger_types}. 
    Moreover, for each type $T$ with precisely $2$ neighbors in $C$ containing more than $3 {k}+6+b+1$ many members, we delete members from $T$ until its size is precisely $3{k}+6+b+1$---the correctness of this step follows from Lemma~\ref{lemma:max_of_type_2} and~\ref{lemma:type2_no_crossing}.

    In the resulting graph, each of the at most $2^k$ many types with at least $3$ neighbors in $C$ has size at most $b+4$, while each of the at most $k^2$ types with precisely $2$ neighbors has size at most $3 {k}+6+b+1$. 
    The kernel bound follows.
\qed    \end{proof}

From \Cref{theorem:vertex_cover}, the runtime bound given by \Cref{theorem:explicit_runtime} and the fact that a vertex cover of size at most $2 \cdot \vc(G)$ can be obtained in linear time, we obtain:

\begin{corollary}
    \textsc{$b$-bend $\beta$-restricted RAC Drawing} is fixed-parameter tractable parameterized by $b+\vc(G)$ , and in particular can be solved in time  $2^{2^{\bigO(\vc(G)+\log b)}} + \bigO(|V(G)|)$.
\end{corollary}

\iflong
\begin{proof}
    Applying the runtime result of $m^{\bigO(m^2)}$ given in \Cref{theorem:explicit_runtime} to the given kernel yields a final runtime of 
    \begin{align*}
        \bigO(b \cdot 2^{\vc(G)})^{\bigO(b^2 4^{\vc(G)})} 
        &= 2^{(\log b +  \vc(G)) \cdot \bigO(b^2 4^{\vc(G)})} \\
        &\leq 2^{b^2\log b \cdot 2^{\bigO(\vc(G))}} \\
        &= 2^{2^{(\log \log b)+2\log b}\cdot 2^{\bigO(\vc(G))}}\\
        &=2^{2^{\bigO(\log b+\vc(G))}},
    \end{align*}
    which concludes the proof.
\qed    \end{proof}
\fi

\ifshort
As our final contribution, we show that the above result also implies fixed-parameter tractability with respect to \textit{neighborhood diversity}. 
This is made possible by the following lemma.
\fi

\iflong
\section{An Extension to Neighborhood Diversity}

We extend the approach used for the vertex cover number to establish fixed-parameter tractability with respect to \textit{neighborhood diversity}. 
Briefly recalling the definition of neighborhood diversity, let two vertices $v,v'$ be of the same type if $N(v)\backslash\{v'\}=N(v')\backslash\{v\}$.

\begin{definition}[\cite{Lampis10,KnopKMT19}]
    The neighborhood diversity $\nd(G)$ of a graph $G$ is the minimum number $k$, such that there exists a partition into $k$ sets, where all vertices in each set have the same type.
\end{definition}

By the definition of neighborhood diversity, each set in the witnessing partition is either an independent set or clique in $G$. 
Edges can occur either on all vertices between two sets or on none (see \Cref{fig:ND_overview}). 
In general, a graph $G$ with neighborhood diversity $\nd(G)$ has a bounded vertex cover number $\vc(G)$. 
Thus, \Cref{theorem:vertex_cover} would already imply tractability of $b$-bend RAC drawings under a bounded neighborhood diversity. 
However,  $\vc(G)$ might be exponentially larger~\cite{Lampis10}. 
For $b$-bend RAC drawable graphs, we can show a better, linear bound on $\vc(G)$.
\fi

\iflong \begin{lemma} \fi
\ifshort \ifreference \begin{lemma}[$\clubsuit$] \else \begin{lemma} \fi \fi
    Let $G$ be a $b$-bend RAC drawable graph with a neighborhood diversity $\nd(G)$. 
    Then $\vc(G) \leq 5 \cdot \nd(G) +b$.
    \label{lemma:ND_VC}
\end{lemma}

\iflong
\begin{proof}
    We begin by showing a linear bound on $\vc(G)$ for $b=0$. 
    Let $S_1,\dots,S_{\nd(G)}$ be a partition witnessing the neighborhood diversity number $\nd(G)$. 
    We build a vertex cover $C$ as follows. 
    The size of each set $S_i$ forming a clique in $G$ is bounded by 5, as a $K_6$ is not straight-line RAC drawable~\cite{DIDIMO20115156}. 
    Put all vertices of such an $S_i$ in $C$. Let $S_i,S_j$ be a pair of two sets, which are both forming an independent set in $G$, and have edges between each other. 
    If there is an edge between a vertex in $S_i$ and a vertex in $S_j$, there is an edge between all vertices of $S_i$ and $S_j$. 
    Let $|S_i| \leq |S_j|$. 
    Recalling that no complete bipartite graph $K_{a,b}$ with $a + b > 7$ and $\min(a,b) > 2$ admits a straight-line RAC drawing~\cite{Didimo}, $S_i \leq 3$. 
    Put $S_i$ into $C$ to cover both sets. 
    In total, we put at most $5 \cdot \nd(G)$ vertices into $C$.

    \begin{figure}
        \centering
        \includegraphics{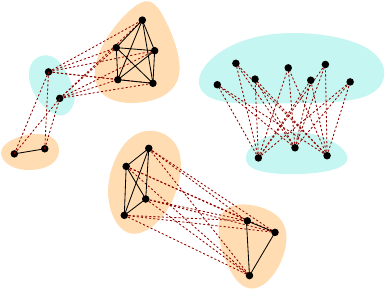}
        \caption{Overview of a graph partitioned into its neighborhood diversity sets. 
        Orange sets build cliques, turquoise sets are independent sets in $G$. 
        Each set may be connected to one ore more other sets.}
        \label{fig:ND_overview}
    \end{figure}

    For arbitrary number of bends $b$, the total number of vertices in clique sets might increase by at most $b$ without making $G$ not $b$-bend RAC drawable. 
    Similarly, the number of vertices in the smaller set $S_i$ of a connected set pair $S_i, S_j$, might increase by at most $b$ over all such sets. 
    So in total, $\vc(G) \leq 5 \cdot \nd(G) +b$.
\qed    \end{proof}

From \Cref{theorem:vertex_cover} and \Cref{lemma:ND_VC} the following theorem follows directly:

\begin{theorem}
    \textsc{$b$-bend $\beta$-restricted RAC Drawing} admits a kernel of size $\bigO(b^2 \cdot \nd(G) \cdot 2^{\nd(G)})$.
\end{theorem}

\begin{corollary}
    \textsc{$b$-bend $\beta$-restricted RAC Drawing} is fixed-parameter tractable parameterized by $\nd(G)+b$, and in particular can be solved in time $2^{b^{\bigO(\nd(G))}} + \bigO(|V(G)|+\nd(G))$.
\end{corollary}
\fi

\section{Concluding Remarks}

We have established the fixed-parameter tractability of \textsc{$b$-bend $\beta$-restricted RAC Drawing} when parameterized by the feedback edge number $\fe(G)$, or by the vertex cover number $\vc(G)$ plus an upper bound $b$ on the total number of bends. 
We have also shown that the latter result implies the fixed-parameter tractability of the problem w.r.t.\ the neighborhood diversity $\nd(G)$ plus $b$.

A next step in the computational study of RAC Drawings would be to consider whether the problem is fixed-parameter tractable w.r.t.\ $\vc(G)$ alone.
Interestingly, a reduction rule for degree-$2$ vertices without a bound on $b$ is the main obstacle towards obtaining such a fixed-parameter algorithm, and dealing with this case seems to be required if one wishes to generalize the result towards fixed-parameter tractability w.r.t.\ \emph{treedepth}~\cite{Nesetril2015} plus $b$.
A different question one may ask is whether the fixed-parameter algorithm w.r.t.\ $\fe(G)$ can be generalized towards the recently introduced parameter \emph{slim tree-cut width}~\cite{GanianK22}, which can be equivalently seen as a local version of the feedback edge number~\cite{BrandCGHK22}. 
A natural long-term goal within this research direction is then to obtain an understanding of the complexity of \BRAC\ w.r.t.\ treewidth~\cite{RobertsonS84}. 
Last but not least, it would be interesting to see whether our fixed-parameter tractability results can be strengthened by obtaining polynomial kernels for the same parameterizations.

\paragraph*{Acknowledgments} 
The authors graciously accept support from the WWTF (Project ICT22-029) and the FWF (Project Y1329) science funds.

\bibliographystyle{splncs04}
\bibliography{literature.bib}
\end{document}